\documentclass{article}
\usepackage{graphicx}

\usepackage{float}

\usepackage{tikz}

\newtheorem{theorem}{Theorem}
\newtheorem{lemma}[theorem]{Lemma}
\newtheorem{corollary}{Corollary}[theorem]

\newtheorem{remark}{Remark}

\newcommand{\qed}{\ifhmode\unskip\nobreak\fi\ifmmode\ifinner
\else\hskip5 pt\fi\fi \hbox{\hskip5 pt
\vrule width4 pt  height6 pt  depth1.5 pt \hskip 1pt }}

\title{Fiedler vector analysis for particular cases of connected graphs} \author{ Daniel Felisberto Tracin\'a Filho}
\author{Daniel Felisberto Tracin\'a Filho\\
P\'os-gradua\c c\~ao em Sistemas e Computa\c c\~ao - Instituto Militar de Engenharia\\ 
danieltracina@ime.eb.br
\and
Claudia Marcela Justel\\
P\'os-gradua\c c\~ao em Sistemas e Computa\c c\~ao - Instituto Militar de Engenharia \\
cjustel@ime.eb.br
}
\date{\ }

\begin{document}
\maketitle

\begin{abstract} 
In this paper, some subclasses of block graphs are considered in order to analyze Fiedler vector of its members.
Two families of block graphs with cliques of fixed size, the block-path and block-starlike graphs, are analyzed. 
Cases A and B of classification for both families were considered, as well as the behavior of the algebraic connectivity when some vertices and edges are added for particular cases of block-path graphs.  

\end{abstract}


\section{Introduction}\label{intro}

We consider  $G=(V,E)$ an undirected, unweighted and simple graph. The sizes of its sets of vertices and edges are $|V| = n$, $|E|=m$. 
The Laplacian matrix of a graph $G$, $L(G)$, is the symmetric and semidefinite positive matrix $D(G)-A(G)$, where $D(G)$ is the diagonal matrix with the degrees of the vertices of $G$ and $A(G)$ is the adjacency matrix of $G$. 
The eigenvalues of $L(G)$, the roots of the characteristic polynomial of $L(G)$, are $n$ non-negative real numbers, and zero is always an eigenvalue of $L(G)$. We denote them as $0 = \lambda_1 \leq \lambda_2 \leq .... \leq \lambda_n$. The second smallest eigenvalue of $L(G)$, $\lambda_2(G)$, is called the algebraic connectivity of $G$. An eigenvector associated with the algebraic connectivity is called a Fiedler vector.
The algebraic connectivity gives some measure of how connected the graph is. Fiedler (\cite{F73}), showed that $\lambda_2(G)  > 0$ if and only if $G$ is connected. Moreover,  the algebraic connectivity of $G$ does not decrease if an edge is added to $G$ and the complete graph has $\lambda_2(K_n) =n$. Also, Fiedler proved that $\lambda_2(G) \leq \kappa(G)$, where $\kappa(G)$ is the vertex connectivity of $G$. 
For a connected graph $G$ and a Fiedler vector $\mathbf{y}$, the $i-th$ entry of $\mathbf{y}$, $y_i$, gives a valuation of vertex $i$. 
When the graph is a tree (acyclic and connected graph), the entries of a Fiedler vector of such tree provide extra information. A tree $T$ is called a Type 1 tree if there is a vertex $z$ such that $y_z = 0$ and is adjacent to a vertex with non zero valuation. In this case $z$ is called characteristic vertex. When there is no entry of $\mathbf{y}$ equal to zero, the tree has an edge $(u,w)$ with $y_u > 0$ and $y_w <0$, and in that case $T$ is called a Type 2 tree and the edge $(u,w)$ the characteristic edge (\cite{GM87}).

Fiedler vectors and algebraic connectivity of graphs have been studied by several authors. For instance some results about trees are presented in \cite{KNS96} and  \cite{P07}, others works such as \cite{KF98} and \cite{KRT15} deal with connected graphs. 
The surveys in \cite{A07} and \cite{K06} provide overviews of the literature on algebraic connectivity.

Even for trees, to identify classes of connected graphs for which its elements are of the same type is not an easy task. For instance, path graphs  $P_n$ are classified by type according to their parity. 
Also, in \cite{P07} Patra shows conditions for a broom tree by of Type 2 and claims that all the broom trees are of Type 2, he also says that will be nice to get a proof of this claim. 

This paper consider two families of connected graphs, both of them belonging to the block graph class, the block-path and the block-starlike graphs. Some properties of Fiedler vectors for graphs in those families are analyzed and conditions for classifying by type such graphs are given.

The paper is organized as follows. In Section~\ref{preliminaries}, we review some basic concepts and results of the literature. Section~\ref{classes} presents the families  block-path and block-starlike graphs, as well as some properties about them. Finally, the conclusions  of the work are presented in Section~\ref{conclusions} and the references in the last section.  

\section{Preliminaries}\label{preliminaries}

In this section, some concepts and results about Graph Theory and Spectral Graph Theory are introduced.

First, let $G=(V,E)$ be a connected graph. The neighborhood of a vertex  $v \in V$ is denoted by $N(v) = \{w \in
V :  (v, w) \in  E\}$ and its closed neighborhood by $N[v] = N(v)  \cup \{v\}$. Two vertices $u, v \in V$ are true twins if $N[u] = N[v]$.

The vertex connectivity of $G$, denoted $\kappa(G)$, is  the minimum number of vertices whose removal from $G$ leaves a disconnected or trivial graph. 

A vertex $v \in V$ is a point of articulation or cutpoint if $G  \backslash v$ is disconnected, and in this case $\kappa(G)=1$. 
A maximal connected subgraph without any points of articulation is called a block. 
The distance, $d_G(u, v)$, in $G$ between two vertices is the length of a shortest path between $u, v$ in $G$. 
%The greatest distance between any two vertices in $G$ is the diameter of $G$, $diam(G)$. 
The eccentricity, $e_G(v)$, of a vertex $v$ in $G$ is its greatest distance from any other vertex. The center of $G$, denoted $center(G)$, is the set of vertices with minimum eccentricity. 
For any $S \subseteq V$, the subgraph of $G$ induced
by $S$ is denoted $G[S]$. If $G[S]$ is a complete subgraph then $S$ is a clique in $G$.

A tree is an acyclic and connected graph. We denote by $P_j$ the path on $j$ vertices and $K_{1,q}$ the star on $q+1$ vertices. 

A tree $T$ is called {\bf starlike} if $T$ is homeomorphic to a star $K_{1,m}$. For $m \geq 3$, 
$T$ has a unique vertex $v$ of degree $m$ and $T  \backslash v$ is a union of $m$ paths (\cite{WS79}).\newline

%{\color{red} TIRAR ?? Given $n \geq q \geq 2$, the {\bf broom} tree with $n$ vertices $T_{n,q}$ is the graph resulting from the %coalescence of one extreme point of $P_{n-q}$ with the center of the star $K_{1,q}$. Broom trees were introduced in [3].} %\newline

A graph is a {\bf  block graph} if every block is a clique. 
A proper interval graph is an interval graph that has an intersection model in which no interval properly contains another. Proper interval graphs are also known as indifference graphs. 
A graph is a block indifference graph when it belongs to the intersection of both classes, indifference and block graphs \cite{BS99}.\newline

Next, some concepts and properties on graph spectra are introduced. The following results consider $G=(V,E)$ a weighted graph with $w(e) \geq 0$ for all $e \in E$ (the weight of edge), and define the Laplacian matrix of $G$, $L(G)$ whose entry $(i,j)$ equals $-w(i,j)$ if $(i,j) \in E$;  zero if $i \not = j$ and $(i,j) \not \in E$; and the sum of weights of edges incident to vertex $i$, if $i=j$.\newline

In \cite{F75} Fiedler proved the following theorem which describes some structure of any Fiedler vector of a connected graph.\newline

\begin{theorem}\label{AB}
(\cite{F75}) Let $G$ be a connected graph and $\mathbf{y}$ a Fiedler vector of $G$. Then exactly one of the following cases occurs: \newline

Case A: There is a single block $C$ in $G$ which contains both positively and negatively valuated vertices. Each other block has either vertices with positive valuation only, vertices with negative valuation only, or vertices with zero valuation only. Every path $P$ which contains at most two points of articulation in each block, which starts in $C$ and contains just one vertex $k$ in $C$ has the property that the valuations at points of articulation contained in $P$ form either an increasing, or decreasing or a zero sequence among this path according to whether $y_k > 0$, $y_k < 0$ or $y_k = 0$; in the last case all vertices in $P$ have valuation zero.

Case B: No block of $G$ contains both positively and negatively valuated vertices. There exists a unique vertex $z$ which has valuation zero and is adjacent to a vertex with a non-zero valuation. This vertex $z$ is a point of articulation. Each block contains either vertices with positive valuation only, with negative valuation only, or with zero valuation only. Every path containing both positively and negatively valuated vertices passes through $z$.\newline
\end{theorem}

Kirkland et al (\cite{KF98}) presented another characterization of graphs for which Case B of Theorem \ref{AB} holds.  
This characterization is based on the concept of Perron components at points of articulation. \newline
For a positive $n \times n$ matrix $M$, the Perron value of $M$, $\rho(M)$ is defined as the maximum eigenvalue of the matrix.
Let $G$ be a connected weighted graph with Laplacian matrix $L(G)$. For a vertex $v$, we denote the connected components of $G  \backslash v$ by  $C_1, ..., C_p$. If $p \geq 2$, $v$ is a point of articulation. For each connected component $C_i$, $1 \leq i \leq p$ let $L(C_i)$ be the principal submatrix of $L(G)$ corresponding to the vertices of $C_i$. The Perron value of $C_i$ is the Perron value of the positive matrix  $L(C_i)^{-1}$, denoted by $\rho (L(C_i)^{-1})$. We say that $C_j$ is the Perron component at $v$ if 
$\rho (L(C_j)^{-1}) = \max_{1 \leq i \leq p}\rho (L(C_i)^{-1})$.\newline

\begin{theorem}\label{kirkB}
(\cite{KF98}) Let $G=(V,E)$ be a weighted graph. Case B of Theorem \ref{AB} holds if and only if there are two or more Perron components at $z$. Further, in that case, the algebraic connectivity $\lambda_2(G)$ is given by $\frac{1}{\rho (L(C)^{-1})}$ for any Perron component $C$ at $z$. If the Perron components at $z$ are $C_1, ..., C_m$, for $1 \leq i \leq m$, let the Perron vector of $L(C_i)^{-1}$ be $x_i$, normalized so that its entries sum to 1. For each $2 \leq i \leq m$, let $b_{i-1}$ be the Fiedler vector which valuates the vertices of $C_1$ by $x_1$, the vertices of $C_i$ by       $-x_i$ and all other vertices zero; then $b_1, ...,  b_{m-1}$ is a basis for the eigenspace corresponding to  $\lambda_2(G)$. In particular, if there are $m$ Perron components at $z$, then the multiplicity of   $\lambda_2(G)$ is $m-1$, and every Fiedler vector has zeros in the positions corresponding to $z$, and to the vertices of the non-Perron components at $z$. Finally, for any vertex $v \not = z$, the unique Perron component at $v$ is the component containing $z$.
\end{theorem}

In other words, Theorem \ref{kirkB} can be rewritten as in Corollary \ref{kirkAB}.\newline

\begin{corollary}\label{kirkAB}
(\cite{KF98}) Let $G=(V,E)$ be a weighted graph. Case A of Theorem \ref{AB} holds if and only if there is a unique Perron component at every vertex of $G$. Case B of Theorem \ref{AB} holds if and only if there is a unique vertex at which there are two or more Perron components. 
\end{corollary}

%\begin{theorem}\label{perroncomp}
%(\cite{KF98}) Suppose that $G$ is a weghted graph with algebraic connectivity $\lambda_2(G)$, and that Case A of Theorem \ref{AB} holds. Let %$\mathbf{y}$  be a Fiedler vector, and $B_0$ be the unique block of $G$ containing both positively and negatively valuated vertices in %$\mathbf{y}$. If we have a point of articulation $v$ of $G$, let $C_0$ denote the set of vertices in the connected component of $G  \backslash %v$ which contains vertices in $B_0$ and let $C_1$ denote the vertices in $G  \backslash C_0$. Permute and partition the Laplacian matrix %as\newline
%
%\begin{center}
%    $\left[ \begin{array}{cc}
%    L(C_1) & A^T \\
%    A & L(C_0)
%    \end{array} \right]$, com 
%    $A = \left[ \begin{array}{cccc}
%    0 & \ldots & 0 & \\
%    \vdots &  & \vdots & -\theta \\ 
%    0 & \ldots & 0 &
%    \end{array} \right]$
%\end{center}
%
%(here vertex $v$ correspond to the last row of $L(C_1)$), and partition $\mathbf{y}$ as $[\mathbf{y_0}, \mathbf{y_1}]^T$. If 
%$y(v) \not = 0$, then
% \begin{center}
%     $L(C_1)^{-1} + \frac{\theta ^Ty_0}{\theta ^T1(\theta ^T 1 y(v) - \theta ^T y_0)}J$
% \end{center}
%is a positive matrix whose Perron value is $\frac{1}{\lambda_2(G)}$ and whose Perron vector is a scalar multiple of $\mathbf{y_1}$. Further, at %every vertex $x$ of $G$, the unique Perron component at $x$ is the component containing vertices in $B_0$.
%\end{theorem}

Observe that for unweighted graphs, we consider each of its edges weight equals 1 in order to apply  Theorem \ref{kirkB} and Corollary \ref{kirkAB}.

\section{{\bf Block-path} and  {\bf block-starlike} classes}\label{classes}

The {\bf block-path} class, denoted  $\mathcal{B}$, is a subclass of block-indifference graphs, with elements denoted $G_{k,p}$, that contain all cliques of size $k \geq 2$ and $p \geq 1$ points of articulation. Moreover, the number of vertices in $G_{k,p}$ is $n = k(p + 1) - p$. For $p=0$, $G_{k,0} = K_k$. The block-path graph $G_{4,3}$ is presented in Figure~\ref{fig:blockpath}.\newline

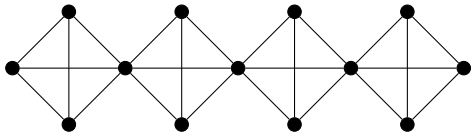
\begin{figure}[H]
\begin{center}
	\begin{tikzpicture}
  [scale=.5,auto=left]
 \tikzstyle{bi}=[
circle, draw, fill=black, 
                        inner sep=1.8pt, minimum width=4pt]
  \tikzstyle{ca}=[
circle, draw,
                        inner sep=1.8pt, minimum width=4pt]
		%%%%%%%%%%%%%%%%  C
	% linha base 
	\node  [bi]  (1) at (5,0) {}; 
	\node  [bi]  (4) at (8,0)  {};
	\node  [bi]  (8) at (11,0) {};
	\node  [bi] (13) at (14,0) {};
	\node  [bi]  (19) at (17,0)  {};
	%   linha inf 
	\node   [bi] (3) at (6.5,-1.5) {};
        \node [bi] (20) at (9.5,-1.5){};
        \node [bi] (21) at (12.5,-1.5){};
        \node [bi] (22) at (15.5,-1.5){};
        %    linha sup 
        \node  [bi] (2) at (6.5,1.5)  {};  
        \node  [bi] (5) at (9.5,1.5)   {}; 
        \node  [bi] (9) at (12.5,1.5)   {}; 
       \node  [bi]  (14) at (15.5,1.5)   {}; 
 \foreach \from/\to in {1/2,1/3,1/4,2/3,2/4,3/4,4/5,5/8,4/8,8/13,13/19,13/14,14/19,20/4,20/5,20/8,21/9,21/8,21/13,22/14,22/13,22/19,9/8,9/13}	
  \draw (\from) -- (\to);
		\end{tikzpicture}
	\caption{$G_{4,3} \in \mathcal{B}$} 
	\label{fig:blockpath}   
	\end{center}
\end{figure}

Let $G = S_{r,k,p_1,...,p_r} \in \mathcal{A}$ with $r \geq 3$ and $v$ its central vertex such that the connected components of $G  \backslash v$ are denoted by 
 $G_{k,p_1} \backslash v$, $G_{k,p_2} \backslash v$, $\ldots$, $G_{k,p_r}  \backslash v$, with $p_1 \geq p_2 \geq \ldots \geq p_r$. $G$ is a {\bf block-starlike}. The class of block-starlike graphs is denoted $\mathcal{A}$.
 
Figure~\ref{fig:blockstarlike} shows $S_{3,4,1,2,3} \in \mathcal{A}$ with $v$ its central vertex and $S_{3,4,1,2,3} \backslash v =  G_{4,1}\backslash v \cup G_{4,2} \backslash v \cup G_{4,3}\backslash v$. Observe that in this example, the central vertex $v$ do not belong to  $center(S_{3,4,1,2,3})$.\newline

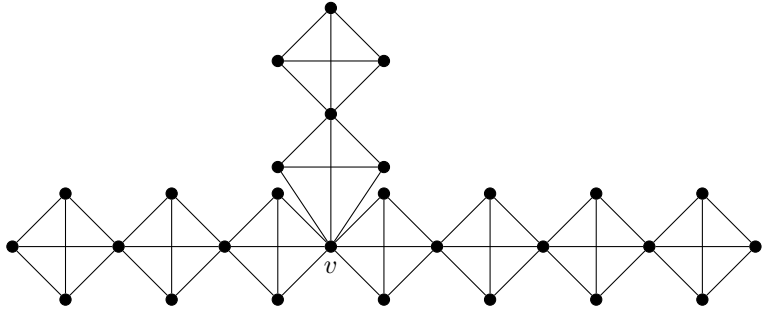
\begin{figure}[H]
	\centering
	\begin{tikzpicture} [scale=2.0,auto=left]
                       %[>=latex',join=bevel,scale=2pt]
	\tikzstyle{selected edge} = [draw,line width=1pt,->,red!30]
        \node (a) at (-50bp,10bp) [draw,circle,inner sep=1.5pt,fill=black!100,label=90:{}] {};
	\node (b) at (-50bp,-10bp) [draw,circle,inner sep=1.5pt,fill=black!100,label=90:{}] {};
	\node (c) at (-60bp,0bp) [draw,circle,inner sep=1.5pt,fill=black!100,label=90:{}] {};
	\node (1) at (-40bp,0bp) [draw,circle,inner sep=1.5pt,fill=black!100,label=90:{}] {};
	\node (2) at (-30bp,10bp) [draw,circle,inner sep=1.5pt,fill=black!100,label=90:{}] {};
	\node (3) at (-30bp,-10bp) [draw,circle,inner sep=1.5pt,fill=black!100,label=90:{}] {};
	\node (4) at (-20bp,0bp) [draw,circle,inner sep=1.5pt,fill=black!100,label=270:{}] {};
	\node (5) at (-10bp,10bp) [draw,circle,inner sep=1.5pt,fill=black!100,label=270:{}] {};
	\node (6) at (-10bp,-10bp) [draw,circle,inner sep=1.5pt,fill=black!100,label=90:{}] {};
	\node (7) at (0bp,45bp) [draw,circle,inner sep=1.5pt,fill=black!100,label=90:{}] {};
	\node (8) at (-10bp,35bp) [draw,circle,inner sep=1.5pt,fill=black!100,label=90:{}] {};
	\node (9) at (10bp,35bp) [draw,circle,inner sep=1.5pt,fill=black!100,label=270:{}] {};
	\node (10) at (0bp,25bp) [draw,circle,inner sep=1.5pt,fill=black!100,label=270:{}] {};
	\node (11) at (-10bp,15bp) [draw,circle,inner sep=1.5pt,fill=black!100,label=90:{}] {};
	\node (12) at (10bp,15bp) [draw,circle,inner sep=1.5pt,fill=black!100,label=90:{}] {};
	\node (13) at (40bp,0bp) [draw,circle,inner sep=1.5pt,fill=black!100,label=90:{}] {};
	\node (14) at (30bp,10bp) [draw,circle,inner sep=1.5pt,fill=black!100,label=270:{}] {};
	\node (15) at (30bp,-10bp) [draw,circle,inner sep=1.5pt,fill=black!100,label=270:{}] {};
	\node (16) at (20bp,0bp) [draw,circle,inner sep=1.5pt,fill=black!100,label=90:{}] {};
	\node (17) at (10bp,-10bp) [draw,circle,inner sep=1.5pt,fill=black!100,label=90:{}] {};
	\node (18) at (10bp,10bp) [draw,circle,inner sep=1.5pt,fill=black!100,label=90:{}] {};
	\node (19) at (0bp,0bp) [draw,circle,inner sep=1.5pt,fill=black!100,label=270:{$v$}] {};
        \node (20) at (50bp,10bp) [draw,circle,inner sep=1.5pt,fill=black!100,label=90:{}] {};
	\node (21) at (50bp,-10bp) [draw,circle,inner sep=1.5pt,fill=black!100,label=90:{}] {};
	\node (22) at (60bp,0bp) [draw,circle,inner sep=1.5pt,fill=black!100,label=270:{}] {};
        \node (23) at (70bp,10bp) [draw,circle,inner sep=1.5pt,fill=black!100,label=90:{}] {};
	\node (24) at (70bp,-10bp) [draw,circle,inner sep=1.5pt,fill=black!100,label=90:{}] {};
	\node (25) at (80bp,0bp) [draw,circle,inner sep=1.5pt,fill=black!100,label=270:{}] {};

% arestas
	\draw [] (1) -- node {} (a);
	\draw [] (1) -- node {} (b);
	\draw [] (1) -- node {} (c);
	\draw [] (a) -- node {} (b);
	\draw [] (a) -- node {} (c);
	\draw [] (b) -- node {} (c);

	\draw [] (1) -- node {} (2);
	\draw [] (2) -- node {} (3);
	\draw [] (3) -- node {} (4);
	\draw [] (4) -- node {} (1);
	\draw [] (4) -- node {} (2);
	\draw [] (1) -- node {} (3);
	
	\draw [] (4) -- node {} (5);
	\draw [] (5) -- node {} (6);
	\draw [] (6) -- node {} (19);
	\draw [] (19) -- node {} (4);
	\draw [] (19) -- node {} (5);
	\draw [] (4) -- node {} (6);
	
	\draw [] (7) -- node {} (8);
	\draw [] (8) -- node {} (10);
	\draw [] (10) -- node {} (9);
	\draw [] (9) -- node {} (7);
	\draw [] (7) -- node {} (10);
	\draw [] (8) -- node {} (9);
	
	\draw [] (10) -- node {} (11);
	\draw [] (11) -- node {} (19);
	\draw [] (12) -- node {} (19);
	\draw [] (10) -- node {} (12);
	\draw [] (11) -- node {} (12);
	\draw [] (10) -- node {} (19);
	
	\draw [] (13) -- node {} (14);
	\draw [] (14) -- node {} (16);
	\draw [] (16) -- node {} (15);
	\draw [] (15) -- node {} (13);
	\draw [] (14) -- node {} (15);
	\draw [] (16) -- node {} (13);
	
	\draw [] (16) -- node {} (17);
	\draw [] (17) -- node {} (19);
	\draw [] (19) -- node {} (18);
	\draw [] (18) -- node {} (16);
	\draw [] (16) -- node {} (19);
	\draw [] (17) -- node {} (18);

        \draw [] (13) -- node {} (20);
	\draw [] (13) -- node {} (21);
	\draw [] (13) -- node {} (22);
	\draw [] (20) -- node {} (21);
	\draw [] (20) -- node {} (22);
	\draw [] (21) -- node {} (22);

        \draw [] (22) -- node {} (23);
	\draw [] (22) -- node {} (24);
	\draw [] (22) -- node {} (25);
	\draw [] (23) -- node {} (24);
	\draw [] (23) -- node {} (25);
	\draw [] (24) -- node {} (25);
	
	\end{tikzpicture}
	\caption{$S_{3,4,1,2,3} \in \mathcal{A}$}
        \label{fig:blockstarlike}
\end{figure}

Next, we present some new results about classes $\mathcal{A}$ and $\mathcal{B}$. \newline

\begin{lemma}\label{twins}
Let $G$ be a block graph with a point of articulation. If $a, b$ are true twins in $G$ and $\mathbf{y}$ is a Fiedler vector, then $y_a = y_b$.
\end{lemma}

\begin{proof}
Let  $|V| = n$, $a, b \in V$ be true twin vertices and the Laplacian matrix of $G$ denoted $L = (l_{i,j})$. As $\lambda_2(G)$ is a  Laplacian eigenvalue of $G$ associated with $\mathbf{y}$, $L \mathbf{y} = \lambda_2(G) \mathbf{y}$ and     
\begin{center}
$\lambda_2(G)(y_a- y_b) = \sum _{s=1} ^n l_{a,s}y_{s} - \sum _{s=1} ^n l_{b,s}y_{s}$
\end{center}

Vertices $a$ and $b$ are true twins then, for all $s \in V \backslash \{a,b\}$ $l_{a,s} = l_{b,s}$. So, 

\begin{center}
$\lambda_2(G)(y_a - y_b) = l_{a,a}y_{a} + l_{a,b}y_{b} - l_{b,a}y_{a} - l_{b,b}y_{b}$

$\lambda_2(G)(y_a - y_b) = (l_{a,a} - l_{b,a})y_{a} + (l_{a,b} - l_{b,b})y_{b}$
\end{center}

Since  $l_{a,a} = l_{b,b}$, 

\begin{center}
$\lambda_2(G)(y_a - y_b) = \lambda_2(G)(y_{a} - y_{b}) = q(y_{a}-y_{b})$, with $q = l_{a,a} - l_{b,a} = deg_G(a) + 1 > 1 \in N$.
\end{center}

However, by hypothesis the graph $G$ has a point of articulation and then $\lambda_2(G) \leq 1$.
So it is not possible to be $y_a \not = y_b$. Then, the result holds. 
\qed
\end{proof}

\begin{theorem}\label{pathB}
Let $G_{k,p} \in  \mathcal{B}$. Then $p$ is an odd number if and only if Case B of Theorem \ref{AB} holds for $G_{k,p}$.
\end{theorem}

\begin{proof}
Label the vertices of  $G_{k,p}$ from 1 to $n= k(p+1) - p$ in the following way: from one block with only 1 point of articulation, all the points of articulation are labelled by $k, 2k-1, 3k-2,..., (p+1) k - p$, and inside each block, the true twin vertices are labeled in order by block with numbers between the labels of the corresponding points of articulation of the block.
Then vertex $v*$ labeled $(\frac{p+1}{2} k - \frac{p+1}{2} + 1)$ is the unique vertex in $center(G_{k,p})$ with odd $p \geq 1$. So, by definition of the block-path class, the graph 
$G_{k,p} - v*$ has two Perron components and by Corollary \ref{kirkAB} it holds if and only if Case B of Theorem \ref{AB} is true for $G_{k,p}$. 
\qed
\end{proof}

\begin{remark}
By Theorem \ref{pathB} we can conclude that $G_{k,p} \in  \mathcal{B}$ with $p$ an even number if and only if Case A of Theorem \ref{AB} holds for  $G_{k,p}$.
\end{remark}

\begin{remark}
If $G = S_{r,k,p_1,...,p_r} \in \mathcal{A}$ with $r \geq 2$ and $p_1 = p_2 = ... = p_r$, then by Corollary \ref{kirkAB} Case B of Theorem \ref{AB} holds for $G$ and the center of the graph equals the central vertex. Figure~\ref{fig:blockstarlikeB} shows $S_{3,4,1,1,1} \in \mathcal{A}$ with $v$ its central vertex.
\end{remark}

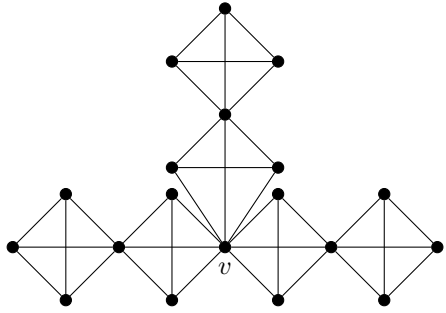
\begin{figure}[H]
	\centering
	\begin{tikzpicture} [scale=2.0,auto=left]
                       %[>=latex',join=bevel,scale=2pt]
	\tikzstyle{selected edge} = [draw,line width=1pt,->,red!30]
	\node (1) at (-40bp,0bp) [draw,circle,inner sep=1.5pt,fill=black!100,label=90:{}] {};
	\node (2) at (-30bp,10bp) [draw,circle,inner sep=1.5pt,fill=black!100,label=90:{}] {};
	\node (3) at (-30bp,-10bp) [draw,circle,inner sep=1.5pt,fill=black!100,label=90:{}] {};
	\node (4) at (-20bp,0bp) [draw,circle,inner sep=1.5pt,fill=black!100,label=270:{}] {};
	\node (5) at (-10bp,10bp) [draw,circle,inner sep=1.5pt,fill=black!100,label=270:{}] {};
	\node (6) at (-10bp,-10bp) [draw,circle,inner sep=1.5pt,fill=black!100,label=90:{}] {};
	\node (7) at (0bp,45bp) [draw,circle,inner sep=1.5pt,fill=black!100,label=90:{}] {};
	\node (8) at (-10bp,35bp) [draw,circle,inner sep=1.5pt,fill=black!100,label=90:{}] {};
	\node (9) at (10bp,35bp) [draw,circle,inner sep=1.5pt,fill=black!100,label=270:{}] {};
	\node (10) at (0bp,25bp) [draw,circle,inner sep=1.5pt,fill=black!100,label=270:{}] {};
	\node (11) at (-10bp,15bp) [draw,circle,inner sep=1.5pt,fill=black!100,label=90:{}] {};
	\node (12) at (10bp,15bp) [draw,circle,inner sep=1.5pt,fill=black!100,label=90:{}] {};
	\node (13) at (40bp,0bp) [draw,circle,inner sep=1.5pt,fill=black!100,label=90:{}] {};
	\node (14) at (30bp,10bp) [draw,circle,inner sep=1.5pt,fill=black!100,label=270:{}] {};
	\node (15) at (30bp,-10bp) [draw,circle,inner sep=1.5pt,fill=black!100,label=270:{}] {};
	\node (16) at (20bp,0bp) [draw,circle,inner sep=1.5pt,fill=black!100,label=90:{}] {};
	\node (17) at (10bp,-10bp) [draw,circle,inner sep=1.5pt,fill=black!100,label=90:{}] {};
	\node (18) at (10bp,10bp) [draw,circle,inner sep=1.5pt,fill=black!100,label=90:{}] {};
	\node (19) at (0bp,0bp) [draw,circle,inner sep=1.5pt,fill=black!100,label=270:{$v$}] {};

% arestas
	\draw [] (1) -- node {} (2);
	\draw [] (2) -- node {} (3);
	\draw [] (3) -- node {} (4);
	\draw [] (4) -- node {} (1);
	\draw [] (4) -- node {} (2);
	\draw [] (1) -- node {} (3);
	
	\draw [] (4) -- node {} (5);
	\draw [] (5) -- node {} (6);
	\draw [] (6) -- node {} (19);
	\draw [] (19) -- node {} (4);
	\draw [] (19) -- node {} (5);
	\draw [] (4) -- node {} (6);
	
	\draw [] (7) -- node {} (8);
	\draw [] (8) -- node {} (10);
	\draw [] (10) -- node {} (9);
	\draw [] (9) -- node {} (7);
	\draw [] (7) -- node {} (10);
	\draw [] (8) -- node {} (9);
	
	\draw [] (10) -- node {} (11);
	\draw [] (11) -- node {} (19);
	\draw [] (12) -- node {} (19);
	\draw [] (10) -- node {} (12);
	\draw [] (11) -- node {} (12);
	\draw [] (10) -- node {} (19);
	
	\draw [] (13) -- node {} (14);
	\draw [] (14) -- node {} (16);
	\draw [] (16) -- node {} (15);
	\draw [] (15) -- node {} (13);
	\draw [] (14) -- node {} (15);
	\draw [] (16) -- node {} (13);
	
	\draw [] (16) -- node {} (17);
	\draw [] (17) -- node {} (19);
	\draw [] (19) -- node {} (18);
	\draw [] (18) -- node {} (16);
	\draw [] (16) -- node {} (19);
	\draw [] (17) -- node {} (18);
	
	\end{tikzpicture}
	\caption{$S_{3,4,1,1,1} \in \mathcal{A}$}
        \label{fig:blockstarlikeB}
\end{figure}

\begin{theorem}\label{starlikeA}
Let $G = S_{r,k,p_1,...,p_r} \in \mathcal{A}$. If $p_2+p_3 +1 \geq p_1$ and $p_1 > p_2$ then Case A of Theorem \ref{AB} holds for $G$. 
\end{theorem}

\begin{proof}
We will prove that for every vertex in $G$ there is only one Perron component. Then Corollary \ref{kirkAB} holds for $G$ in the condition of the hypothesis.
Observe that is sufficient to consider points of articulations of $G$.
Let $L=L(G)$ and  $v$ be the central vertex of $G$, then the connected components of $G  \backslash v$ are $G_{k,p_1}  \backslash v$, $G_{k,p_2}  \backslash v$, \ldots, $G_{k,p_r}  \backslash v$ and for each $ 2 \leq i \leq r$, $G_{k,p_i}  \backslash v$ is isomorphic to a proper subgraph of $G_{k,p_1}  \backslash v$, and $\rho(L(G_{k, p_1}  \backslash v)^{-1}) >\rho(L(G_{k, p_i}  \backslash v)^{-1}) $. 
So, $\rho(L_v^{-1}) = \rho(L(G_{k, p_1}  \backslash v)^{-1})$  and  $G_{k,p_1}  \backslash v$ is the unique Perron component of $G  \backslash v$.\newline
Let $w \in V$,  $w \not = v$. We consider two cases:\newline 
1. $G \backslash w = K_{k-1} \cup   G' \backslash w$ with $G'$ a block starlike, and \newline 
2. $G \backslash w =  G_{k,a} \backslash w \cup H \backslash w$, where $H$ is a block starlike containing as subgraph a block path graph with $p_2 + p_3 + 1$ articulation points  and exists $1 \leq i \leq 3$ such that $1 < a \leq p_i -1$. \newline
When case 1 holds there is only one Perron component of $G$ in $w$. \newline
On the other hand, if case 2 holds, consider the following situations: \newline 
If $i=1,  p_1 -1 \geq a$. Since by hypothesis $p_2+p_3 +1 \geq p_1$, we can conclude that $p_2 + p_3 \geq a$. Moreover, as  $H$ contains at least 2 connected components of $G  \backslash w$ and $p_2 + p_3 \geq a$ then  $G_{k,a}$ is isomorphic to a subgraph of $H$. 

If $i \not =1$, since $p_i -1 \geq a$ then $a \leq p_i - 1 < p_1 $ and then  $G_{k,a}$ is isomorphic to a subgraph of $H$.

Thus, by Corollary \ref{kirkAB} the theorem is true.\newline
\qed
\end{proof}

Grone and Merris (\cite{GM87}) showed that the algebraic connectivity of a Type 1 tree $T$ does not change when a degree one vertex adjacent to the characteristic vertex is added to $T$. The next theorem generalizes that result for graphs  $G_{k,p} \in  \mathcal{B}$ with $p$ an odd number. \newline
 
Observe that if $p$ is odd, then  $center(G_{k,p}) = u$ where $u$ is the point of articulation labeled  $(\frac{p+1}{2} k - \frac{p+1}{2} + 1)$  as in proof of  Theorem \ref{pathB}.\newline

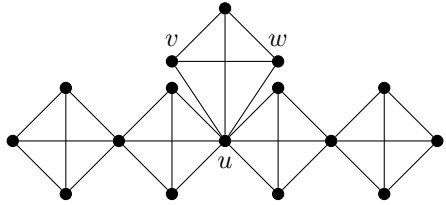
\begin{figure}[H]
	\centering
	\begin{tikzpicture} [scale=2.0,auto=left]
                       %[>=latex',join=bevel,scale=2pt]
	\tikzstyle{selected edge} = [draw,line width=1pt,->,red!30]
	\node (1) at (-40bp,0bp) [draw,circle,inner sep=1.5pt,fill=black!100,label=90:{}] {};
	\node (2) at (-30bp,10bp) [draw,circle,inner sep=1.5pt,fill=black!100,label=90:{}] {};
	\node (3) at (-30bp,-10bp) [draw,circle,inner sep=1.5pt,fill=black!100,label=90:{}] {};
	\node (4) at (-20bp,0bp) [draw,circle,inner sep=1.5pt,fill=black!100,label=270:{}] {};
	\node (5) at (-10bp,10bp) [draw,circle,inner sep=1.5pt,fill=black!100,label=270:{}] {};
	\node (6) at (-10bp,-10bp) [draw,circle,inner sep=1.5pt,fill=black!100,label=90:{}] {};
	%\node (7) at (0bp,45bp) [draw,circle,inner sep=1.5pt,fill=black!100,label=90:{}] {};
	%\node (8) at (-10bp,35bp) [draw,circle,inner sep=1.5pt,fill=black!100,label=90:{}] {};
	%\node (9) at (10bp,35bp) [draw,circle,inner sep=1.5pt,fill=black!100,label=270:{}] {};
	\node (10) at (0bp,25bp) [draw,circle,inner sep=1.5pt,fill=black!100,label=270:{}] {};
	\node (11) at (-10bp,15bp) [draw,circle,inner sep=1.5pt,fill=black!100,label=90:{$v$}] {};
	\node (12) at (10bp,15bp) [draw,circle,inner sep=1.5pt,fill=black!100,label=90:{$w$}] {};
	\node (13) at (40bp,0bp) [draw,circle,inner sep=1.5pt,fill=black!100,label=90:{}] {};
	\node (14) at (30bp,10bp) [draw,circle,inner sep=1.5pt,fill=black!100,label=270:{}] {};
	\node (15) at (30bp,-10bp) [draw,circle,inner sep=1.5pt,fill=black!100,label=270:{}] {};
	\node (16) at (20bp,0bp) [draw,circle,inner sep=1.5pt,fill=black!100,label=90:{}] {};
	\node (17) at (10bp,-10bp) [draw,circle,inner sep=1.5pt,fill=black!100,label=90:{}] {};
	\node (18) at (10bp,10bp) [draw,circle,inner sep=1.5pt,fill=black!100,label=90:{}] {};
	\node (19) at (0bp,0bp) [draw,circle,inner sep=1.5pt,fill=black!100,label=270:{$u$}] {};

% arestas
	\draw [] (1) -- node {} (2);
	\draw [] (2) -- node {} (3);
	\draw [] (3) -- node {} (4);
	\draw [] (4) -- node {} (1);
	\draw [] (4) -- node {} (2);
	\draw [] (1) -- node {} (3);
	
	\draw [] (4) -- node {} (5);
	\draw [] (5) -- node {} (6);
	\draw [] (6) -- node {} (19);
	\draw [] (19) -- node {} (4);
	\draw [] (19) -- node {} (5);
	\draw [] (4) -- node {} (6);
	
	%\draw [] (7) -- node {} (8);
	%\draw [] (8) -- node {} (10);
	%\draw [] (10) -- node {} (9);
	%\draw [] (9) -- node {} (7);
	%\draw [] (7) -- node {} (10);
	%\draw [] (8) -- node {} (9);
	
	\draw [] (10) -- node {} (11);
	\draw [] (11) -- node {} (19);
	\draw [] (12) -- node {} (19);
	\draw [] (10) -- node {} (12);
	\draw [] (11) -- node {} (12);
	\draw [] (10) -- node {} (19);
	
	\draw [] (13) -- node {} (14);
	\draw [] (14) -- node {} (16);
	\draw [] (16) -- node {} (15);
	\draw [] (15) -- node {} (13);
	\draw [] (14) -- node {} (15);
	\draw [] (16) -- node {} (13);
	
	\draw [] (16) -- node {} (17);
	\draw [] (17) -- node {} (19);
	\draw [] (19) -- node {} (18);
	\draw [] (18) -- node {} (16);
	\draw [] (16) -- node {} (19);
	\draw [] (17) -- node {} (18);
	
	\end{tikzpicture}
	\caption{$G'$, obtained from $G=G_{4,3}$ and $G_{4,0}$, with $\lambda_2(G)=\lambda_2(G') = 0.32938$}
        \label{fig:exteo36}
\end{figure}

\begin{theorem}\label{pathGM}
Let $G = G_{k,p} \in \mathcal{B}$ with $n$ vertices and $p$ an odd number. Let  $\mathbf{y}$ be a Fiedler vector of $G$ and $u\in V(G)$ the unique vertex in $center(G)$. Let $G'$ be the graph obtained as the coalescence between vertex $u$ of $G$ with a vertex of $G_{k,0}$. Let  $\mathbf{y}'$ be a  $(n+k-1)$-dimensional vector such that $y'_v = y_v, \forall v \in V (G)$ and $y'_v = 0$ otherwise. Then  $G' \in \mathcal{A}$, $\mathbf{y}'$ is a Fiedler vector of $G'$ and $\lambda_2(G') = \lambda_2(G)$.
\end{theorem}

\begin{proof}
Let $v, w$ be a pair of vertices different from $u$ in the new block $K_{k}$ of $G'$ (as in Figure~\ref{fig:exteo36}). Since $v, w$ are twin vertices,  by Lemma \ref{twins}  $y_v = y_w$. Then 
\begin{center}
     $\sum ^n _{s=1} l_{v,s}y_{s} = (k+1) y_v - k y_v - y_u = y_v - y_u = \lambda_2(G') y_v$
     
     $ y_v = \frac{1}{1-\lambda_2(G')}y_u \quad (*)$
 \end{center}
As $p \geq 1$, $G'\backslash u$ has 3 connected components, 2 of them are isomorphic and contain at least one block $K_k$ and the other one is isomorphic to $K_{k-1}$. Then, by Theorem \ref{kirkB} case B holds for $G'$ and $y_u=0$. From equation $(*)$, we have that  $y'_v=0$ for all twin  vertices in the new block  $K_{k}$ of $G'$, and the result follows. \newline
\qed
\end{proof}

%An alternative proof for Theorem \ref{pathGM}3.5 is given next. \newline
%
%\begin{proof}
% Let 
%\begin{center}
%    $L(G') = \left[ \begin{array}{c|c}
%    L(G)_{n \times n} + k E_{u,u} & M^T \\ \hline
%    M & ((k+1)I - J)_{k \times k} 
%    \end{array} \right]$, where  
%    $M = \left[ \begin{array}{ccccc}
%     0& \ldots & -1 & \ldots & 0\\
%    \vdots &  & -1 &          & \vdots \\ 
%    0 & \ldots &  -1& \ldots & 0
%    \end{array} \right]_{k \times n}$ and 
%$E_{u,u} = $ $n \times n$ matrix with all zero elements unless $E_{u,u}=1 $ .
%\end{center}
%And 
%\begin{center}
%    $L(G') \mathbf{y}' = \left[ \begin{array}{c}
%    (L(G) +  k E_{u,u}) \mathbf{y}_{n \times 1}  \\ \hline
%    \mathbf{0}_{k-1 \times 1}
%    \end{array} \right] = 
%\lambda_2(G) \left[ \begin{array}{c}
%\mathbf{y}_{n \times 1}  \\ \hline
%    \mathbf{0}_{k-1 \times 1}
%  \end{array} \right].
%$
%\end{center}
%Since $y_u=0$, we can conclude that $\mathbf{y'}$ is an eigenvector associated with the eigenvalue $\lambda_2(G')$. 
%And as $G'$ is a subgraph of $G$, $\lambda_2(G)$ is the algebraic connectivity of $G'$ with eigenvector $\mathbf{y'}$.\newline
%\end{proof}

\section{Conclusions}\label{conclusions}

Two subclasses of block graphs were considered in this paper: block-path and block-starlike graphs.
First, we showed some properties of Fiedler vectors associated with those graphs, such as identifying whether Cases A or B of Theorem \ref{AB} holds for block-path graphs. Then we exhibited particular conditions to have block-starlike graphs for which Case A of the same theorem holds. 
Next, we obtained the necessary conditions to mantain the value of the algebraic connectivity of a block-path graph when some vertices and edges are added. 
We plan on exploring other subclasses of block graphs in the future.

\section*{Acknowledgments}
This work was supported by CAPES, Coordena\c c\~ao de Aperfei\c coamento de Pessoal do N\'ivel Superior - C\'odigo de Financiamento 001.

\end{document}